\documentclass[number,3p,twocolumn,preprint]{elsarticle}
\usepackage[active]{srcltx}
\usepackage{amsthm,mathrsfs}
\usepackage{amssymb,bbm,bm,amsmath}
\usepackage{hyperref}
\usepackage{xargs}
\usepackage{xifthen}
\usepackage[matrix,frame,arrow]{xy}
\usepackage{amsmath}

\newcommand{\qw}[1][-1]{\ar @{-} [0,#1]}
\newcommand{\gate}[1]{*{\xy *+<.6em>{#1};p\save+LU;+RU **\dir{-}\restore\save+RU;+RD **\dir{-}\restore\save+RD;+LD **\dir{-}\restore\POS+LD;+LU **\dir{-}\endxy} \qw}
\newcommand{\measureD}[1]{*{\xy*+=+<.5em>{\vphantom{\rule{0em}{.1em}#1}}*\cir{r_l};p\save*!R{#1} \restore\save+UC;+UC-<.5em,0em>*!R{\hphantom{#1}}+L **\dir{-} \restore\save+DC;+DC-<.5em,0em>*!R{\hphantom{#1}}+L **\dir{-} \restore\POS+UC-<.5em,0em>*!R{\hphantom{#1}}+L;+DC-<.5em,0em>*!R{\hphantom{#1}}+L **\dir{-} \endxy} \qw}
\newcommand{\multimeasureD}[2]{*+<1em,.9em>{\hphantom{#2}}\save[0,0].[#1,0];p\save !C *{#2},p+LU+<0em,0em>;+RU+<-.8em,0em> **\dir{-}\restore\save +LD;+LU **\dir{-}\restore\save +LD;+RD-<.8em,0em> **\dir{-} \restore\save +RD+<0em,.8em>;+RU-<0em,.8em> **\dir{-} \restore \POS !UR*!UR{\cir<.9em>{r_d}};!DR*!DR{\cir<.9em>{d_l}}\restore \qw}
\newcommand{\multigate}[2]{*+<1em,.9em>{\hphantom{#2}} \qw \POS[0,0].[#1,0];p !C *{#2},p \save+LU;+RU **\dir{-}\restore\save+RU;+RD **\dir{-}\restore\save+RD;+LD **\dir{-}\restore\save+LD;+LU **\dir{-}\restore}
\newcommand{\ghost}[1]{*+<1em,.9em>{\hphantom{#1}} \qw}
\newcommand{\Qcircuit}[1][0em]{\xymatrix @*=<#1>}
\newcommand{\pureghost}[1]{*+<1em,.9em>{\hphantom{#1}}}
\newcommand{\multiprepareC}[2]{*+<1em,.9em>{\hphantom{#2}}\save[0,0].[#1,0];p\save !C
  *{#2},p+RU+<0em,0em>;+LU+<+.8em,0em> **\dir{-}\restore\save +RD;+RU **\dir{-}\restore\save
  +RD;+LD+<.8em,0em> **\dir{-} \restore\save +LD+<0em,.8em>;+LU-<0em,.8em> **\dir{-} \restore \POS
  !UL*!UL{\cir<.9em>{u_r}};!DL*!DL{\cir<.9em>{l_u}}\restore}
\newcommand{\prepareC}[1]{*{\xy*+=+<.5em>{\vphantom{#1\rule{0em}{.1em}}}*\cir{l^r};p\save*!L{#1} \restore\save+UC;+UC+<.5em,0em>*!L{\hphantom{#1}}+R **\dir{-} \restore\save+DC;+DC+<.5em,0em>*!L{\hphantom{#1}}+R **\dir{-} \restore\POS+UC+<.5em,0em>*!L{\hphantom{#1}}+R;+DC+<.5em,0em>*!L{\hphantom{#1}}+R **\dir{-} \endxy}}
\newcommand{\poloFantasmaCn}[1]{{{}^{#1}_{\phantom{#1}}}}
\usepackage{color}
\usepackage{framed}
\definecolor{shadecolor}{rgb}{.725,.725,.725}
\theoremstyle{definition}
\newtheorem{definition}{Definition}
\newtheorem*{definition*}{Definition}
\theoremstyle{plain}

\newtheorem{proposition}{Proposition}

\newtheorem*{corollary*}{Corollary}

\newcommand\State[1]{|#1)}

\newcommand\Effect[1]{(#1|}

\newcommand{\RBraKet}[2]{(#1|#2)}

\newcommand\SetStates{\mathrm{St}}

\newcommand\SetEffects{\mathrm{Eff}}

\newcommand\SetTransf{\mathrm{Transf}}

\newcommand\Test[1]{\bm{\mathcal #1}}

\newcommand\Transformation[1]{\mathcal{#1}}

\newcommandx\NSState[5][usedefault, addprefix=\global, 2=, 3=, 4=,
5=]{#1_{#2\ifthenelse{\NOT\isempty{#3}}{\,#3}{}#4\ifthenelse{\NOT\isempty{#5}}{,{#5}}
{}}}

\newcommandx\NSEffect[5][usedefault, addprefix=\global, 2=, 3=, 4=,
5=]{#1_{#2\ifthenelse{\NOT\isempty{#3}}{\,#3}{}#4\ifthenelse{\NOT\isempty{#5}}{,#5}{}}}

\newcommandx\SState[5][usedefault, addprefix=\global, 2=, 3=, 4=,
5=]{\State{\NSState{#1}[#2][#3][#4][#5]}}

\newcommandx\SEffect[5][usedefault, addprefix=\global, 2=, 3=, 4=, 5=]{
\Effect{\NSEffect{#1}[#2][#3][#4][#5]} }

\newcommand\ExtrTransformation[4]{\Transformation{F}_{#1\,#2}^{#3\,#4}}

\newcommandx\ProbCond[4][usedefault, addprefix=\global, 1=,
2=]{\operatorname{Pr}_{#1}^{#2}\![#3|#4]}

\newcommand{\ustickcool}[1]{*!D!<0em,-.1em>=<0em>{\scriptstyle #1}}

\def\sys#1{{\mathrm #1}}

\def\sn{\sys{n}}
\def\sm{\sys{m}}
\def\sp{\sys{p}}
\def\sq{\sys{q}}
\def\sA{\sys{A}}
\def\sB{\sys{B}}
\def\sC{\sys{C}}
\def\sD{\sys{D}}
\def\sI{\sys{I}}

\begin{document}
\title{Determinism without causality}
\author[quit,infn]{Giacomo Mauro D'Ariano\fnref{url}}
\ead{dariano@unipv.it}
\author[quit]{Franco Manessi\fnref{url}}
\ead{franco.manessi01@ateneopv.it}
\author[quit,infn]{Paolo Perinotti\fnref{url}} 
\ead{paolo.perinotti@unipv.it}
\address[quit]{QUIT group,  Dipartimento di Fisica, via Bassi   6, 27100 Pavia, Italy.} 
\address[infn]{INFN Gruppo IV, Sezione di Pavia, via Bassi, 6, 27100 Pavia, Italy.} 
\date{\today}
\fntext[url]{URL: \url{http://www.quantummechanics.it}}

\begin{abstract}
Causality has been often confused with the notion of determinism. It is mandatory to separate the
two notions in view of the debate about quantum foundations. Quantum theory provides an example of
causal not-deterministic theory. Here we introduce a toy operational theory that is deterministic
and  non-causal, thus proving that the two notions of causality and determinism are totally
independent. 
\end{abstract}

\maketitle

\section{Introduction}
Causality is subject of a very extensive literature, encompassing hundreds of contemporary books and
technical articles. It hits a wide spectrum of disciplines, ranging from pure philosophy to law,
economics, natural sciences, and, in particular, physics. Perhaps the most natural connection with
physics is in philosophy, from the early work of Aristotle, to the cornerstone of Ren\'ee Descartes,
who broke the ground for the modern view of David Hume and Immanuel Kant, up to the contemporary
works on physical causation of Wesley Salmon \cite{salmon1998cauality} and Phil Dowe
\cite{dowe2007physical}.

The recent reconsideration of foundations of physics, with particular focus on quantum theory, has
brought research in theoretical physics to explore issues in the territory shared with philosophy
and epistemology.  A paradigmatic case is the issue of realism raised by the founding fathers von
Neumann \cite{vonNeumann1932} and Einstein \cite{EPR} in regards of completeness of quantum theory.

The problem of causality has remained in the realm of philosophy, and stayed 
only in
the background of physics, without the status of a physical law or the rank of a principle. Most of
the time causality creeps in the form of {\em ad hoc} assumptions based on empirical
evidences---like the discard of advanced potentials in electrodynamics or the Kramers-Kronig
relations---or it is part of the interpretation of the theory---e.g.\ in special relativity---or
else it is hidden in the theoretical framework, as in Hardy axiomatization of quantum theory
\cite{Hardy:2001jk}.

A notion that is traditionally connected with causality in physics and philosophy is {\em
  determinism}, which is deeply entangled with causality, to the extent that the two are often
merged into the {\em causal determinism}, or even confused, as in the exemplar quotation from Max
Planck: \emph{``An event is causally determined if it can be predicted with certainty''}
\cite{planck1941kausalbegriff}. This confusion between the two notions is the source of the common
misleading way of regarding quantum correlations as "spooky action at a distance"---the commonplace
of perfect EPR correlations interpreted as causation.

The notion of determinism arose within the clockwork-universe vision of classical mechanics,
assessing that the state of a system at an initial time completely determines the state at any later
time. Classical mechanics, however, identifies the {\em state} (the point in the phase-space) with
the {\em measurement-outcome}, while the two notions are radically different in quantum theory, and
more generally in {\em operational probabilistic theories} \cite{QUIT-ProbTheories,QUIT-Arxiv}.
These allow us to define determinism outside the framework of classical mechanics which is already
deterministic, avoiding the confusion between state and measurement-outcome.  In a probabilistic
context \cite{QUIT-ProbTheories} determinism is identified with the property of a theory of having
all probabilities of physical events equal to either zero or one---a definition which has no causal
connotation. 

The property of causality within classical theory is trivialized by the irrelevance of the notion of
measurement, which is identified with that of state itself. Complementarity is the feature that
breaks the classical identification between observation and preparation (measurement and state).
Causality is the independence of the probability of preparation from the choice of observation: this
definition of causality distills all the intuitive guises in which it appears in physics, with an
intimate relation with the Einsteinian notion. In this formulation it is the first axiom of quantum
theory in the derivation of Ref.~\cite{QUIT-Arxiv}.

Quantum theory provides a relevant example of operational probabilistic theory that is causal and
not deterministic. In this paper we introduce a toy theory that is deterministic and non-causal. The
purpose is to prove in this way that neither causality implies determinism, nor determinism implies
causality, namely the two notions are logically independent.  In the concluding section we will
further discuss about the relation between the definition of causality of Sect.
\ref{sec:prob-theories} and the customary problem of physical causation along with the cause-effect
connection.

\section{Review on Operational Probabilistic 
Theories}\label{sec:prob-theories}
Before starting we need to review the basic definitions and notations for Operational Probabilistic
Theories (OPT). For a detailed discussion see \cite{QUIT-ProbTheories}.  The basic notion in the
operational framework is that of {\em test}. A test $\Test{A}=\{\Transformation{A}_{i}\}$ describes
an elementary operation which generally produces the readout of an outcome $i$, heralding the
occurrence of an {\em event} $\Transformation{A}_{i}$. Tests are also specified by an input and an
output label, e.g.\ $\sA,\sB$, which identify the {\em system types} ({\em systems}, for short).
The test \( \Test{A} \) and its building events \( \Transformation{A}_{i} \in \Test{A} \) can be
represented by means of boxes as \( \Qcircuit @C=1em @R=.7em @! R { & \ustickcool{\sA}\qw &
  \gate{\Test{A}} & \ustickcool{\sB}\qw &\qw} \) and \( \Qcircuit @C=1em @R=.7em @! R { &
  \ustickcool{\sA}\qw & \gate{\Transformation{A}_{i}} & \ustickcool{\sB}\qw &\qw} \) respectively.
The role of labelling input and output systems is to provide rules for connecting tests in
sequences: an output wire labeled $\sA$ can be connected only to an input wire with the same label
$\sA$.  Notice that the input/output relation has no causal connotation, and does not entail an
underlying ``time arrow''. Here ``input/output'' has to be understood as a functional dependence,
namely the relation that links the variable $x$ to the function evaluation \( f(x) \). As it will be
clear shortly, only in a causal theory it is possible to understand the input/output relation as a
time-arrow.

The event \( \Transformation{B}_{j} \circ \Transformation{A}_{i} \) belonging to the sequential
composition \( \Test{B}\circ\Test{A} \) of the tests \( \Test{A} \) and \( \Test{B} \) is
represented as \( \Qcircuit @C=1em @R=.7em @! R {& \ustickcool{\sA}\qw & \gate
  {\Transformation{A}_{i}} & \ustickcool{\sB}\qw &
  \gate{\Transformation{B}_{j}}&\ustickcool{\sC}\qw&\qw} \) (a similar graphical representation
holds also for the test \( \Test{B}\circ\Test{A} \) itself).  For every system $\sA$ there exists a
unique singleton test $\{\Transformation{I}_{\sA} \}$ such that $\Transformation{I}_{\sB}
\circ\Transformation{A}=\Transformation{A}\circ\Transformation{I}_{\sA}$ for every event
$\Transformation{A}$ with input $\sA$ and output $\sB$. For every couple of systems $(\sA,\sB)$ we
can form the composite system $\sC:=\sA\sB$, on which we can perform tests $\Test A\otimes\Test B$
with events $\Transformation{A}_i\otimes\Transformation{B}_j$ in {\em parallel composition}
represented as follows
\begin{equation*}
  \begin{aligned}
    \Qcircuit @C=1em @R=.7em @! R {& \qw \poloFantasmaCn \sA & \multigate{1} {A_i\otimes B_j} & \qw
      \poloFantasmaCn \sB &\qw\\
      & \qw \poloFantasmaCn \sC & \ghost {A_i\otimes B_j} & \qw \poloFantasmaCn \sD &\qw}
  \end{aligned}=
  \begin{aligned}
    \Qcircuit @C=1em @R=.7em @! R {& \qw \poloFantasmaCn \sA & \gate {A_i} & \qw \poloFantasmaCn \sB
    &\qw\\ & \qw \poloFantasmaCn \sC & \gate {B_j} & \qw \poloFantasmaCn \sD &\qw}
  \end{aligned}
\end{equation*}
and satisfying the following condition: \[ ( \Transformation{C}_{h} \otimes \Transformation{D}_{k} )
\circ ( \Transformation{A}_{i} \otimes \Transformation{B}_{j} ) = ( \Transformation{C}_{h} \circ
\Transformation{A}_{i} ) \otimes ( \Transformation{D}_{k} \circ \Transformation{B}_{j} ). \] Notice
that here $\otimes$ is a formal symbol for parallel composition, and not the usual tensor product of
linear spaces.  There is a special system type $\sI$, the {\em trivial system}, such that
$\sA\sI=\sI\sA=\sA$. The tests with input system $\sI$ and output $\sA$ are called {\em
  preparation-tests} of $\sA$, while the tests with input system $\sA$ and output $\sI$ are called
{\em observation-tests} of $\sA$.  Preparation-events of $\sA$ are denoted by the symbols \(
\State{\rho}_\sA \) or \( \Qcircuit @C=.5em @R=.5em { \prepareC{\rho} & \ustickcool{\sA} \qw & \qw }
\), and observation-events by \( \Effect{ c}_\sA \) or \( \Qcircuit @C=.5em @R=.5em { &
  \ustickcool{\sA} \qw & \measureD{ c} } \).  Note that the words {\em preparation-test} and {\em
  observation-test} have an intrinsic causal connotation (usually one observes something that has
been prepared, and not viceversa), however, here the two words should be taken only as technical
terms. The two terms recover their usual meaning in a causal theory--the commonly studied case--and
our abuse of terminology is for the sake of limiting temporary technical words.

An arbitrary complex test obtained by parallel and sequential composition of box diagrams is called
\emph{circuit}. 
\begin{figure}
  \begin{equation*}
    \begin{aligned}
    \Qcircuit @C=1em @R=.7em @! R {
      \multiprepareC{3}{\Psi_{i_1}}&\qw\poloFantasmaCn{\sA}&\multigate{1}{\Transformation{A}_{i_2}}&\qw\poloFantasmaCn{\sB}&\gate{\Transformation{C}_{i_4}}&\qw\poloFantasmaCn{\sC}&\multigate{1}{\Transformation{E}_{i_6}}&\qw\poloFantasmaCn{\sD}&\multimeasureD{2}{G_{i_8}}\\
      \pureghost{\Psi_{i_1}}&\qw\poloFantasmaCn{\sys E}&\ghost{\Transformation{A}_{i_2}}&\qw\poloFantasmaCn{\sys F}&\multigate{1}{\Transformation{D}_{i_5}}&\qw\poloFantasmaCn{\sys G}&\ghost{\Transformation{E}_{i_6}}&&\pureghost{G_{i_8}}\\
      \pureghost{\Psi_{i_1}}&\qw\poloFantasmaCn{\sys H}&\multigate{1}{\Transformation{B}_{i_3}}&\qw\poloFantasmaCn{\sys L}&\ghost{\Transformation{D}_{i_5}}&\qw\poloFantasmaCn{\sys M}&\multigate{1}{\Transformation{F}_{i_7}}&\qw\poloFantasmaCn{\sys N}&\ghost{G_{i_8}}\\
      \pureghost{\Psi_{i_1}}&\qw\poloFantasmaCn{\sys O}&\ghost{\Transformation{B}_{i_3}}&\qw\poloFantasmaCn{\sys P}&\qw&\qw &\ghost{\Transformation{F}_{i_7}}\\
    }
    \end{aligned}
  \end{equation*}
  \caption{ The closed circuit in the figure represent the joint probability \(
    \ProbCond{i_1,i_2,\dots i_8}{\bm\Psi,\Test A,\dots,\bm G} \) of outcomes $i_1,i_2,\dots i_8$
    conditioned by the choice of tests ${\bm\Psi,\Test A,\dots,\bm G}$. Since the output of the event \(
    \Transformation{A}_{i_2} \) is connected to the input of the event \( \Transformation{D}_{i_5} \)
    through the system \( \sys F \), the event \( \Transformation{A}_{i_2} \) immediately precedes the
    event \( \Transformation{D}_{i_5} \) (\( \Transformation{A}_{i_2} \prec_1 \Transformation{D}_{i_5}
    \)).  Similarly, since between the event \( \Transformation{B}_{i_3} \) and the event \(
    \Transformation{E}_{i_6} \) there is \( \Transformation{D}_{i_5} \) such that \(
    \Transformation{B}_{i_3} \prec_1 \Transformation{D}_{i_5}\prec_1 \Transformation{E}_{i_6}  \), the
    event \( \Transformation{B}_{i_3} \) precedes the event \( \Transformation{E}_{i_6} \) (\(
    \Transformation{B}_{i_3} \prec \Transformation{E}_{i_6} \)). If the closed circuit of the figure
    belongs to a causal theory, we have e.g.~that the marginal probability of the event \(
    \Transformation{D}_{i_5}\in\Test{D} \) cannot depend on the choice of any test \( \Test{X} \) such
    that \( \Test{X} \not\prec \Test{D} \), i.e.~\( \ProbCond{ i_5 }{ \bm\Psi, \Test{A}, \Test{B},
    \Test{C}, \Test{D}, \Test{E}, \Test{F}, \bm G } = \ProbCond{ i_5 }{ \bm\Psi,\Test{A}, \Test{B} } \).
    \label{fig:closed-circuit}}
\end{figure}
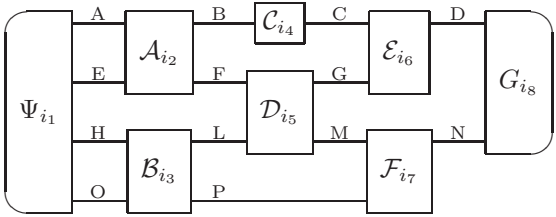
A circuit is \emph{closed} if its overall input and output systems are the trivial ones. Figure
\ref{fig:closed-circuit} is an example of closed circuit.  Given a circuit we say that an event
$\Transformation{H}$ {\em is immediately connected to the input of} $\Transformation{K}$, and write
$\Transformation{H}\prec_1\Transformation{K}$, if there is an output system of $\Transformation{H}$
that is connected with an input system of $\Transformation{K}$; e.g.~referring to the circuit in
Fig.~\ref{fig:closed-circuit} \( \Transformation{A}_{i_2} \prec_1 \Transformation{D}_{i_5} \). We
can moreover introduce the transitive closure $\prec$ of the relation $\prec_1$, and we say that
$\Transformation{H}$ {\em is connected to the input of} $\Transformation{K}$ if
$\Transformation{H}\prec\Transformation{K}$ (e.g.~\( \Transformation{B}_{i_3} \prec
\Transformation{E}_{i_6} \). The two relations \( \prec_1 \) and \( \prec \) can be trivially
extended from events to tests.

A theory is {\em probabilistic} if every closed circuit represents a probability distribution;
e.g.~the closed circuit in Fig.~\ref{fig:closed-circuit} represents the probability \(
\ProbCond{i_1,i_2,\dots i_8}{\bm\Psi,\Test A,\dots,\bm G} \) of outcomes $i_1,i_2,\dots i_8$
conditioned by the choice of tests ${\bm\Psi,\Test A,\dots,\bm G}$ \footnote{To be more precise the
definition of probabilistic theory includes also the following formal rule for the composition of
events of trivial systems $p_i\otimes p_j:=p_i p_j=:p_i\circ p_j $, stating the independence of
closed circuits.}. In probabilistic theories we can quotient the set of preparation-events of $\sA$
by the equivalence relation \( \State{\rho}_\sA\sim\State{\sigma}_\sA \Leftrightarrow \) \emph{the
probability of preparing} \( \State{\rho}_\sA \) \emph{and measuring} \( \Effect{c}_\sA \) \emph{is
the same as preparing} \( \State{\sigma}_\sA \) \emph{and measuring} \( \Effect{c}_\sA \) \emph{for
every observation-event} \( \Effect{c}_\sA \) \emph{of} \( \sA \) (and similarly for
observation-events).  The equivalence classes of preparation-events and observation-events of \(
\sA\) will be denoted by the same symbols as their elements $\State{\rho}_\sA$ and $\Effect{c}_\sA$,
respectively, and will be called \emph{state} $\State{\rho}_\sA$ for system $\sA$, and \emph{effect}
$\Effect{c}_\sA$ for system $\sA$.  For every system $\sA$, we will denote by $\SetStates(\sA)$,
$\SetEffects(\sA)$ the sets of states and effects, respectively.  States and effects are real-valued
functionals on each other, and then they can be naturally embedded in reciprocally dual real vector
spaces, \( \SetStates_{\mathbb{R}}(\sA) \) and \( \SetEffects_{\mathbb{R}}(\sA) \), whose dimension
$D_\sA$ is assumed here to be finite.  The application of the effect $\Effect{c_i}_\sA$ on the state
$\State{\rho}_\sA$ is written as $\RBraKet{ c_i }{ \rho }_\sA$ and corresponds to the closed circuit
$\Qcircuit @C=.5em @R=.5em { \prepareC{\rho} & \ustickcool{\sA} \qw & \measureD{ c_i } } $, denoting
therefore the probability of the $i$-th outcome of the observation-test $\bm
c=\{\Effect{c_i}_\sA\}_{i\in\eta}$ performed on the state $\rho$ of system $\sA$, i.e.~$ \RBraKet{
c_i }{ \rho }_\sA :=\ProbCond{c_i}{\bm\rho}$.

Any event with input system $\sA$ and output system $\sB$ induces a collection of linear mappings
from \( \SetStates_{\mathbb R}(\sA\sC) \) to \( \SetStates_{\mathbb R}(\sB\sC) \), for varying
system $\sC$. Such a collection is called {\em transformation} from $\sA$ to $\sB$. The set of
transformations from $\sA$ to $\sB$ will be denoted by $\SetTransf(\sA,\sB)$, and its linear span by
$\SetTransf_\mathbb{R}(\sA,\sB)$. The symbols \( \Transformation{A} \) and $ \Qcircuit @C=.5em
@R=.5em { & \ustickcool{\sA} \qw & \gate{\Transformation{A}} & \ustickcool{\sB} \qw &\qw} $ denoting
the event $\Transformation{A}$ will be also used to represent the corresponding transformation.

We now introduce a precise notion of determinism through the following definition
\cite{QUIT-ProbTheories}
\begin{definition}[ODT]\label{def:deterministic-OPT}
  An {\em Operational Deterministic Theory} (ODT) is an OPT with all closed circuits having
  probabilities 0 or 1.
\end{definition}
One cannot forbid the construction of the ``statistical'' version of an ODT (as it happens for
classical mechanics) by considering the OPT which is the convex closure of the ODT. 

Given a set $\mathrm{S}$ the convex cone $\lambda\mathrm{S}$ is the conic hull of $\mathrm{S}$,
namely the set of all conic combinations of elements of \(\mathrm{S} \).  With obvious notation we
have the cones $\lambda\SetStates(\sA)$, $\lambda\SetEffects(\sA)$, and
$\lambda\SetTransf(\sA,\sB)$.  The elements on the extremal rays of the cones are called {\em
atomic}. In the following, we will use the Greek letters to denote states and Latin letters to
denote effects.  Moreover, in the rest of the paper we will not specify the system when it is clear
from the context or it is generic.

An event $\Transformation{A}$ is {\em deterministic} if it belongs to a singleton test.  We will
denote respectively with \( \SetStates_1(\sA) \) , \( \SetEffects_1(\sA) \) and \(
\SetTransf_1(\sA,\sB) \) the set of deterministic states, effects and transformations for systems
$\sA$ and $\sB$, and we will often use the symbols \( \State{\varepsilon} \) and \( \Effect{e} \) to
refer respectively to a deterministic state and effect.  Note that in convex OPTs the sets \(
\SetStates_1(\sA) \) and \( \SetEffects_1(\sA) \) are convex. Deterministic transformations are also
called {\em channels}. 

Among the properties of OPTs, a relevant one is \emph{Local Discriminability}
\cite{QUIT-ProbTheories}, namely the possibility to discriminate multipartite states only through
local measurement on the subsystems:
\begin{definition}[Local Discriminability]\label{def:local-discriminability} 
  If $\State{\rho}_{\sA\sB}, \State{\sigma}_{\sA\sB} \in \SetStates_1(\sA\sB)$ are states and
  $\State{\rho}_{\sA\sB} \neq \State{\sigma}_{\sA\sB}$, then there are two effects $\Effect{a}_\sA \in
  \SetEffects(\sA)$ and $\Effect{b}_\sB \in \SetEffects(\sB)$ such that 
  \begin{equation*}\label{eq:local-discriminability}
    \begin{aligned} 
    \Qcircuit @C=1em @R=.7em @! R {\multiprepareC{1}{\rho}& \qw \poloFantasmaCn
      \sA &\measureD a \\
      \pureghost\rho & \qw \poloFantasmaCn \sB &\measureD b}
    \end{aligned}
    ~\neq~
    \begin{aligned}
    \Qcircuit @C=1em @R=.7em @! R {\multiprepareC{1}{\sigma}& \qw
      \poloFantasmaCn \sA &\measureD a \\
      \pureghost\sigma & \qw \poloFantasmaCn \sB &\measureD b}
    \end{aligned}.
  \end{equation*}
\end{definition}
Local Discriminability is equivalent to $\SetStates_\mathbb{R}(\sA\sB) = \SetStates_\mathbb{R}(\sA)
\otimes \SetStates_\mathbb{R}(\sB)$ \cite{QUIT-Arxiv}, where now the symbol $\otimes$ denotes the
usual tensor product of linear spaces.  The analog condition also holds for the effects.  An
important consequence of Local Discriminability is that a transformation $ \Transformation{T} \in
\SetTransf( \sA,\sB )$ is completely specified by its action on $\SetStates(\sA)$
\cite{QUIT-ProbTheories}:
\begin{equation*}
  \Transformation{C} \State{\rho} = \Transformation{C^\prime} \State{\rho} \quad \forall \State{\rho}
  \in\SetStates(\sA) \ \Rightarrow\ \Transformation{C} = \Transformation{C^\prime}.
\end{equation*}

We now introduce the definition of causality \cite{QUIT-Arxiv}.
\begin{definition}[Causal OPT]\label{def:causality}
  An OPT is \emph{causal} if the probability for every preparation-test $\bm
  \rho=\{\State{\rho_i}\}_{i \in \eta}$ and any two observation-tests $\bm a=\{\Effect{a_j}\}_{j \in
  \chi}$ and $\bm b=\{\Effect{b_j}\}_{j \in \xi}$ one has \( \sum_{j \in \chi} \RBraKet{a_j}{\rho_i} =
  \sum_{k \in \xi} \RBraKet{b_k} {\rho_i},\forall i\in\eta\), namely the probability of the
  preparation is independent of the choice of observation.
\end{definition}
Causality is equivalent to \emph{no backward signaling}
\cite{PhilosophyQuantumInformationEntanglement2}, namely within a closed circuit, the marginal
probability of outcomes for a given test $\Test{H}$ do not depend on the choice of any test $\Test
K$ not connected to the input of $\Test H$, i.e.~$\Test K\not\prec\Test H$.  For example, in the
circuit of Fig.~\ref{fig:closed-circuit} causality implies that
\begin{equation*}
  \ProbCond{ i_5 }{ \bm\Psi, \Test{A}, \Test{B}, \Test{C}, \Test{D}, \Test{E}, \Test{F}, \bm G } =
  \ProbCond{ i_5 }{ \bm\Psi,\Test{A}, \Test{B} } 
\end{equation*}
The present notion of causality is nothing but a rigorous definition of the so-called {\em Einstein
causality}.  Indeed, a corollary of \emph{no backward signaling} is the {\em no-signaling without
interaction} \cite{QUIT-ProbTheories}. A crucial equivalent condition for causality of an OPT is the
uniqueness of the deterministic effect \cite{QUIT-ProbTheories}.

  The possibility of reversing the causal arrow (by defining \emph{backward causality} or
  \emph{retro-causality} as independence of observation on preparation) does not add anything new
  conceptually, since there is an isomorphism between any retro-causal theory and a causal one, upon
  exchanging the roles of input and output.

In the following we will take Local Discriminability for granted. We say that a linear map \(
\Transformation{T}\in\SetTransf_{\mathbb R}(\sA,\sB) \) is {\em admissible} if it locally preserves
the set of states \( \SetStates(\sA\sC) \), namely \(
\Transformation{T}\otimes\Transformation{I}_{\sC}(\SetStates(\sA\sC)) \subseteq \SetStates(\sB\sC)
\). In the following we will assume that every admissible map actually belongs to \(
\SetTransf(\sA,\sB) \).  We will refer to this last assumption as \emph{No-Restriction Hypothesis}
\footnote{In previous literature \cite{QUIT-ProbTheories} the same nomenclature has been used for
  the cone duality $\lambda\SetEffects(\sA)= \lambda\SetStates(\sA)^*$, which is a different
  concept.}.

\section{The deterministic noncausal theory}\label{uffa}
We now introduce an example of non-causal deterministic theory. The systems will be denoted by the
symbols \( \sn\triangleright\sm \), where \( \sn \), \( \sm \) are positive integer numbers, and
they enjoys the property that \( \dim\SetStates_\mathbb{R}(\sn\triangleright\sm) =
\dim\SetEffects_\mathbb{R}(\sn\triangleright\sm) = n \cdot m \). Composition of systems is defined
as \((\sn\triangleright\sm)(\sn^\prime\triangleright\sm^\prime):=\sys{x}\triangleright\sys{y}\),
where \( x = n \cdot n^\prime \) and \( y = m \cdot m^\prime \), consistently with Local
Discriminability.  Notice that this definition is consistent with
associativity and commutativity of parallel composition, as well as
the existence of a trivial system $\sI:=(\sn\triangleright\sm)$ with $n=m=1$. 

Denote by \( \Gamma_n \) the set of all the non-negative integer numbers less than $n$, i.e.~\(
\Gamma_n:=\{0,\ldots,n-1\} \).  The set of states of the system \( \sn\triangleright\sm \) is
defined as \( \SetStates(\sn\triangleright\sm) := \{ \SState{\alpha}[][][f][\Xi] \mid f:
\Xi\to\Gamma_m \text{ and }\Xi \subseteq \Gamma_n\} \). The atomic states of \(
\SetStates(\sn\triangleright\sm) \) are the elements \( \SState{\alpha}[][][f][\{ i \}] \) with \(
f: \{i\}\to\Gamma_m \), \( i\in\Gamma_n \). In the following we will use a special notation for the
atomic states: \(\SState{\alpha}[i][j] := \SState{\alpha}[][][f][\{ i \}] \) with \( f(i)=j \).  The
number of different atomic states for \( \sn\triangleright\sm \) is \( n \cdot m \), i.e.~the same
as the dimension of \( \SetStates_\mathbb{R}(\sn\triangleright\sm) \).  For \(
\Xi,\Upsilon\subset\Gamma_n \) with \( \Xi \cap \Upsilon = \emptyset \), the states of \(
\sn\triangleright\sm \) enjoy the property \( \SState{\alpha}[][][f][\Xi] +
\SState{\alpha}[][][g][\Upsilon] \equiv \SState{\alpha}[][][h][\Xi\cup\Upsilon] \), with \( h:\ \Xi
\cup \Upsilon \to \Gamma_m \), $ h(i):= f(i)$ for $i\in\Xi$, and $ h(i):= g(i)$ for $i\in\Upsilon$.
Notice that for \( \Xi \cap \Upsilon \ne \emptyset \), \( \SState{\alpha}[][][f][\Xi] +
\SState{\alpha}[][][g][\Upsilon] \) is not a valid state. We have that a deterministic state is an
element \( \SState{\varepsilon}[][][f] := \SState{\alpha}[][][f][\Gamma_n] \), hence the set of the
deterministic states is \( \SetStates_1(\sn\triangleright\sm) = \{ \SState{\varepsilon}[][][f],\ f:\
\Gamma_n\to\Gamma_m \} \).

The set of states \( \SetStates(\sys x\triangleright\sys y) \) for the bipartite system \(
\sys{x}\triangleright\sys{y} =(\sn\triangleright\sm)(\sn^\prime\triangleright\sm^\prime)\) is built
up via the definition of bipartite atomic states as parallel composition of single-system atomic
states \(\SState{\alpha}[(s,s^\prime)][(t,t^\prime)]:=
\SState{\alpha}[s][t]\otimes\SState{\alpha}[s^\prime][t^\prime] \), with $\Gamma_x:=
\Gamma_n\times\Gamma_{n^\prime}$ and $\Gamma_y:=\Gamma_m\times\Gamma_{m^\prime}$. It can be 
shown that this is the only possible definition of atomic state consistent with Local Discriminability
(see Props.~1,
%\ref{prop:otimes-is-atomic}, 
and 2
%\ref{prop:unavoidable}, 
in the Appendix).

Under the No-Restriction Hypothesis we can easily build the set of effects for the system \(
\sn\triangleright\sm \) from the set \( \SetStates(\sn\triangleright\sm) \). The atomic effects are
the elements \( \SEffect{a}[s][s^\prime] \) such that \(
\RBraKet{\NSEffect{a}[s][s^\prime]}{\NSState{\alpha}[t][t^\prime]} = \delta_{st} \delta_{s^\prime
  t^\prime} \) (see Prop.~4
%\ref{prop:atomic-effects} 
in the Appendix).
In general, it can be shown that \( \SetEffects(\sn\triangleright\sm) := \{ \SEffect{a}[][][v][E]
\mid v \in \Gamma_n \text{ and } E \subseteq\Gamma_m \} \), using the definition \(
\SEffect{a}[][][v][E \cup F]:= \SEffect{a}[][][v][E] + \SEffect{a}[][][v][F]\) for \( E \cap F =
\emptyset \) (see Prop.~5
%\ref{prop:effects} 
in the Appendix). The atomic effects are \( \SEffect{a}[][][s][\{ s^\prime
\}]\equiv\SEffect{a}[s][s^\prime] \).  The deterministic effects are the elements \( \SEffect{e}[v]
:= \SEffect{a}[][][v][\Gamma_m] \), and one can verify that \(
\RBraKet{\NSEffect{e}[][][v]}{\NSState{\varepsilon}[][][f]} = 1 \) for every \(
\SState{\varepsilon}[][][f]\in\SetStates_1(\sn\triangleright\sm) \). Indeed, one can check that \(
\RBraKet{ \NSEffect{a}[][][v][E] }{ \NSState{\alpha}[][][f][\Xi] } := \chi_\Xi(v) \chi_E(f(v)) \),
with \( \chi_S \) the indicator function of the set \( S \), showing that for \( E = \Gamma_m \), \(
\Xi = \Gamma_n \)---i.e.~for deterministic states and effects---\(
\RBraKet{\NSEffect{e}[][][v]}{\NSState{\varepsilon}[][][f]} = \RBraKet{
  \NSEffect{a}[][][v][\Gamma_m] }{ \NSState{\alpha}[][][f][\Gamma_n] } = 1 \). Notice that for a
generic system \( \sn\triangleright\sm \) there are $n$ different deterministic effects; since an
OPT is causal if and only if for every system there is just a single deterministic effect
\cite{QUIT-ProbTheories}, we conclude that the presented theory is non-causal.

To complete the theory, we need to specify all possible transformations. The set of transformations
\( \SetTransf(\sn\triangleright\sm,\sp\triangleright\sq) \) is built up starting from the
atomic elements \( \ExtrTransformation{s}{s^\prime}{t}{t^\prime} \) with \( (s,s^\prime,t,t^\prime) \in
\Gamma_n\times\Gamma_m\times\Gamma_p\times\Gamma_q \) defined as \(
\ExtrTransformation{s}{s^\prime}{t}{t^\prime} \SState{\alpha}[v][v^\prime] := \delta_s^v
\delta_{s^\prime}^{v^\prime} \SState{\alpha}[t][t^\prime] \) (see
Props.~6, 7, 8, and 9
%\ref{prop:atomic-are-admissible}, \ref{prop:atomic-are-li}, \ref{prop:atomic-are-atomic}, and
%\ref{prop:no-other-atomic} 
in the Appendix).  The other transformations belonging to \(
\SetTransf(\sn\triangleright\sm,\sp\triangleright\sq) \) are the elements \(
\Transformation{T}^{f\,g}_{\Omega} := \sum_{(s^\prime,t) \in \Omega}
\ExtrTransformation{f(t)}{s^\prime}{t}{g(t,s^\prime)} \) with \( \Omega\subseteq
\Gamma_p\times\Gamma_m \), \( f: \Gamma_p\to\Gamma_n \), and \( g: \Gamma_p\times\Gamma_m\to\Gamma_q
\) (see Props.~10, and 11
%\ref{prop:transf-are-admissible}, and \ref{prop:transformations} 
in the Appendix).  Notice that \( \ExtrTransformation{s}{s^\prime}{t}{t^\prime} \equiv
\Transformation{T}^{f\,g}_{\{ s^\prime \}\times\{ t \}} \) with \( f(t) = s \) and \( g(t,s^\prime)
= t^\prime \). The channels from \( \sn\triangleright\sm \) to \( \sp\triangleright\sq \) are the
elements \( \Transformation{T}^{f\,g} := \Transformation{T}^{f\,g}_{\Gamma_n\times\Gamma_p} \).
This completes the construction of the full theory, which is deterministic and non causal.

\bigskip

We can give now an explicit example which shows the non-causal features of the presented theory. Let
us consider a simple case with the system \( 2\triangleright2 \) and the experimenter Alice. Alice
wants to prepare the system \( 2\triangleright2 \) by means of the preparation test \( \{
\SState{\alpha}[][][f][\Xi_i] \}_{i=0,1} \), with \( \Xi_i := \{ i \} \) for \( i=0,1 \), and \( f
\) arbitrary function from \( \Gamma_2 \) to \( \Gamma_2 \). She subsequently measures the system
chosing one observation test between \( \Test{D}_0 := \{ \SEffect{a}[][][0][\Xi_i] \}_{i=0,1} \) and
\( \Test{D}_1 := \{ \SEffect{a}[][][1][\Xi_i] \}_{i=0,1} \). It can be easily seen that the
probability of preparing the state \( \SState{\alpha}[][][f][\Xi_i] \) depends on which observation
Alice wants to perform. Indeed,
\begin{align*}
  & \ProbCond{\NSState{\alpha}[][][f][\Xi_0]}{\Test{D}_0} = \RBraKet{ \NSEffect{a}[][][0][\Xi_0] }{
    \NSState{\alpha}[][][f][\Xi_0] } + \RBraKet{ \NSEffect{a}[][][0][\Xi_1] }{
    \NSState{\alpha}[][][f][\Xi_0] } = \\
  &= \RBraKet{ \NSEffect{e}[0] }{ \NSEffect{\alpha}[][][f][\Xi_0] } =  1,\\
  & \ProbCond{\NSState{\alpha}[][][f][\Xi_0]}{\Test{D}_1} = \RBraKet{ \NSEffect{a}[][][1][\Xi_0] }{
    \NSState{\alpha}[][][f][\Xi_0] } + \RBraKet{ \NSEffect{a}[][][1][\Xi_1] }{
    \NSState{\alpha}[][][f][\Xi_0] } = \\
  &= \RBraKet{ \NSEffect{e}[1] }{ \NSEffect{\alpha}[][][f][\Xi_0] } =  0,
\end{align*}
and similarly for the state \( \SState{\alpha}[][][f][\Xi_1] \). 

We can moreover show how this deterministic non-causal theory violates the no-signalling without
interaction, i.e.~by means of a bipartite deterministic state an experimenter Bob can communicate
with Alice just with local measurements on his own subsystem. Let us suppose that both the systems
of Alice and Bob are \( 2\triangleright2 \), and that they share the bipartite deterministic state
\( \SState{\varepsilon} \in \SetStates(4\triangleright4) \).  Keeping the same notation of the
previous example, let us suppose that Bob can perform the two observation-test \( \Test{D}_0 \), \(
\Test{D}_1 \). It can be easily shown that, unlike in Quantum Theory, if \(
\SState{\varepsilon}_{AB} \) is properly chosen the state Alice sees in her subsystem without
knowing the outcome of the measurement performed by Bob (the so-called marginal state of Alice),
will depend on the choice made by Bob. In this way Alice performing a local observation on her own
subsystem can assess the choice of the measurement made on the other subsystem, getting therefore
information from Bob.  Indeed, if Bob performs the test \( \Test{D}_0 \) the marginal state of Alice
will be 
\begin{equation*}
  \begin{aligned} 
    \Qcircuit @C=1em @R=.7em @! R {\multiprepareC{1}{\varepsilon} & \qw
        \poloFantasmaCn{2\triangleright2} & \qw \\
      \pureghost\varepsilon & \qw \poloFantasmaCn{2\triangleright2} &
        \measureD{\NSEffect{a}[][][0][\Xi_0]}}
  \end{aligned} 
  + 
  \begin{aligned} 
    \Qcircuit @C=1em @R=.7em @! R {\multiprepareC{1}{\varepsilon}& \qw
        \poloFantasmaCn{2\triangleright2} & \qw \\
      \pureghost\varepsilon & \qw \poloFantasmaCn{2\triangleright2}
        & \measureD{\NSEffect{a}[][][0][\Xi_1]}}
  \end{aligned} 
  \equiv
  \begin{aligned} 
    \Qcircuit @C=1em @R=.7em @! R {\multiprepareC{1}{\varepsilon}& \qw
        \poloFantasmaCn{2\triangleright2} & \qw \\
      \pureghost\varepsilon & \qw \poloFantasmaCn{2\triangleright2}
        & \measureD{\NSEffect{e}[0]}}.
  \end{aligned} 
\end{equation*}
Let us choose as deterministic bipartite input the state \( \SState{\varepsilon} := \sum_{st} 
\SState{\alpha}[s][t]\otimes\SState{\alpha}[t][s] \). We have that the marginal state
of Alice when Bob performs the test \( \Test{D}_0 \) is 
\begin{align*}
  & \SEffect{e}[0]_B \SState{\varepsilon}_{AB} = \sum_{st} \SState{\alpha}[s][t] \otimes \RBraKet{ \NSEffect{e}[0] }{
    \NSState{\alpha}[t][s] } = \\
  & = \sum_{st} \chi_{\Xi_t}(0)\
    \chi_{\Gamma_2}(s)\ \SState{\alpha}[s][t] = \sum_{s}
    \SState{\alpha}[s][0],
\end{align*}
namely the deterministic state \( \SState{\varepsilon}[h_0] \in \SetStates(2\triangleright2) \) where \( h_0 \)
is the function such that \( h_0(x):=0 \) \( \forall x \in \Gamma_2 \). Similarly, the marginal
state of Alice when Bob performs the test \( \Test{D}_1 \) is \( \SState{\varepsilon}[h_1] \)---with
\( h_1(x):=1 \) \(\forall x\in\Gamma_2\). Alice can distinguish between the two marginal states \(
\SState{\varepsilon}[h_0] \), \( \SState{\varepsilon}[h_1] \) by means of the test \( \Test{D}_0 \),
assessing the choice of Bob.

The presented deterministic non-causal theory can also be built in a constructive way
\cite{preparation}. It is done in two steps. The first one consists in building a non-causal OPT
through the addition of a non-causal shell around an internal causal core corresponding to the
classical OPT, thanks to a construction analogous to that of quantum combs in the case of quantum
theory \cite{supermaps, comblong}.  Then the resulting two-shell theory is constrained to be
deterministic.An interesting result is that every transformation of the probabilistic non-causal
core+shell theory can be implemented just using elements of the core causal theory
\cite{preparation}.

\section{Discussion}
In summary, in this paper we have established the logical independence of the two notions of
causality and determinism, which play a crucial role in physics, and stay at the core of the debate
about foundations of quantum theory and relativity. As a legacy of classical physics the two
concepts have been often merged into a single one which is unfit to quantum theory, thus leading to
misconceptions. Here by determinism we simply mean that the probabilities of all events are either 0
or 1, whereas for causality we mean the usual Einstein's notion, namely no-signalling from the
future. We have proved that not only the two notions are formally independent, but also this
distinction is not vacuous, since there are indeed both counterexamples of a theory which is causal
and non deterministic and, reversely, a theory which is deterministic and non causal. Quantum theory
provides the first example, while the second one has been introduced in the present paper
(retro-causal deterministic theories are not an good example of non-causal deterministic theory, due
to the mentioned isomorphism between causal and retro-causal theories).

In conclusion of this paper, we want to comment about the relation between the notion of causality
in Definition \ref{def:causality} and the usual cause-effect relation and/or the physical causation
in the philosophical literature, e. g. as in Refs. \cite{salmon1998cauality,dowe2007physical}. After
centuries of debates it may be said with a degree confidence that an empirical notion of
causality/causation is missing, and in all cases the cause-effect connection is of conterfactual
nature. Causality should be always regarded as meaningless outside a theory. Within a theory
Definition \ref{def:causality} is the minimal requirement for the use of the term causality, and only if a
theory is causal it makes sense to identify cause and effect, whereas in a non causal theory the two
words are nonsensical. Obviously within Definition \ref{def:causality} the preparation plays the
role of the cause and the observation that of the effect. Finally, we want to stress that our definition
is exactly the Einsteinian one in special relativity theory.

\section*{Appendix}
In this Appendix we present all the technical results which ensure the consistency of the non-causal
deterministic theory presented in Sect. \ref{uffa}.

\begin{proposition}\label{prop:otimes-is-atomic}
  If Local Discriminability holds, the parallel composition of two atomic transformations is atomic.
\end{proposition}
\begin{proof}
  Let \( \Transformation{A} \in \SetTransf(\sA,\sA^\prime) \), \(
  \Transformation{B} \in \SetTransf(\sB,\sB^\prime) \) be atomic
  transformations between systems. Let us consider the transformation
  \( \Transformation{A} \otimes \Transformation{B} \in
  \SetTransf(\sA\sB,\sA^\prime\sB^\prime) \). Let us suppose that \(
  \Transformation{A} \otimes \Transformation{B} \) can be decomposed
  as follows
  \begin{equation*}
    \begin{aligned}
    \begin{aligned}
      \Qcircuit @C=1em @R=.7em @! R { 
        & \ustickcool{\sA}\qw & \gate{\Transformation{A}} & \ustickcool{\sA^\prime}\qw &\qw \\
        & \ustickcool{\sB}\qw & \gate{\Transformation{B}} & \ustickcool{\sB^\prime}\qw &\qw }
    \end{aligned}
    =
    \begin{aligned}
      \Qcircuit @C=1em @R=.7em @! R { 
        & \ustickcool{\sA}\qw & \multigate{1}{\Transformation{C}} & \ustickcool{\sA^\prime}\qw &\qw \\
        & \ustickcool{\sB}\qw & \ghost{\Transformation{C}} & \ustickcool{\sB^\prime}\qw & \qw }
    \end{aligned}
    +
    \begin{aligned}
      \Qcircuit @C=1em @R=.7em @! R { 
        & \ustickcool{\sA}\qw & \multigate{1}{\Transformation{D}} & \ustickcool{\sA^\prime}\qw &\qw \\
        & \ustickcool{\sB}\qw & \ghost{\Transformation{D}} & \ustickcool{\sB^\prime}\qw &\qw }
    \end{aligned},
\end{aligned}
  \end{equation*}
  for a non trivial couple of transformations \( 0\neq\Transformation{C}, \Transformation{D} \in
  \SetTransf(\sA\sB,\sA^\prime\sB^\prime) \). For any state \( \State{\beta} \in \SetStates(\sB) \),
  and any effect \( \Effect{b} \in \SetEffects(\sB^\prime) \) such that \( \Effect{b}
  \Transformation{B} \State{\beta} \ne 0 \), we have
  \begin{align*}
    & \begin{aligned}
      \Qcircuit @C=1em @R=.7em @! R { 
        & \ustickcool{\sA}\qw & \gate{\Transformation{A}} & \ustickcool{\sA^\prime}\qw & \qw \\
        & \prepareC{\beta} & \gate{\Transformation{B}}  & \measureD{b} }
    \end{aligned}
    =
    \begin{aligned}
      \Qcircuit @C=1em @R=.7em @! R { 
        & \ustickcool{\sA}\qw & \multigate{1}{\Transformation{C}} & \ustickcool{\sA^\prime}\qw &\qw \\
        & \prepareC{\beta} & \ghost{\Transformation{C}} & \measureD{b} }
    \end{aligned}
    + \\ & +
    \begin{aligned}
      \Qcircuit @C=1em @R=.7em @! R { 
        & \ustickcool{\sA}\qw & \multigate{1}{\Transformation{D}} & \ustickcool{\sA^\prime}\qw &\qw \\
        & \prepareC{\beta} & \ghost{\Transformation{D}} & \measureD{b} }
    \end{aligned},
  \end{align*}
  Since the transformation \( \Transformation{A} \) is atomic we have that the transformations \(
  \Effect{b}_\sB \Transformation{C} \State{\beta}_{\sB^\prime}, \Effect{b}_\sB \Transformation{D}
  \State{\beta}_{\sB^\prime} \in \SetTransf(\sA,\sA^\prime)\ \) must be proportional to \(
  \Transformation{A} \); in particular for any state \( \State{\alpha} \in \SetStates(\sA) \), and any
  effect \( \Effect{a} \in \SetEffects(\sA^\prime) \) such that \( \Effect{a} \Transformation{A}
  \State{\alpha} \ne 0 \), it must be
  \begin{align}
    \label{eq:atom-A-C}
    \begin{aligned}
      \Qcircuit @C=1em @R=.7em @! R { 
        & \prepareC{\alpha} & \multigate{1}{\Transformation{C}} & \measureD{a} \\
        & \prepareC{\beta} & \ghost{\Transformation{C}} & \measureD{b} }
    \end{aligned}
    = 
    \mu_{b\beta}^{\Transformation{C}} \Effect{a}
    \Transformation{A} \State{\alpha},
    \\
    \label{eq:atom-A-D}
    \begin{aligned}
      \Qcircuit @C=1em @R=.7em @! R { 
        & \prepareC{\alpha} & \multigate{1}{\Transformation{D}} & \measureD{a} \\
        & \prepareC{\beta} & \ghost{\Transformation{D}} & \measureD{b} }
    \end{aligned}
    = 
    \mu_{b\beta}^{\Transformation{D}} \Effect{a}
    \Transformation{A} \State{\alpha}, 
  \end{align}
  where \( \mu_{b\beta}^{\Transformation{C}} \), \( \mu_{b\beta}^{\Transformation{D}} \) are 
  constants which can depend on the choice of \( \State{\beta} \) and \( \Effect{b} \).  One can repeat
  a similar argument on the other subsystem, getting:
  \begin{align}
    \label{eq:atom-B-C}
    \begin{aligned}
      \Qcircuit @C=1em @R=.7em @! R { 
        & \prepareC{\alpha} & \multigate{1}{\Transformation{C}} & \measureD{a} \\
        & \prepareC{\beta} & \ghost{\Transformation{C}} & \measureD{b} }
    \end{aligned}
    = 
    \lambda_{a\alpha}^{\Transformation{C}} \Effect{b}
    \Transformation{B} \State{\beta},
    \\
    \label{eq:atom-B-D}
    \begin{aligned}
      \Qcircuit @C=1em @R=.7em @! R { 
        & \prepareC{\alpha} & \multigate{1}{\Transformation{D}} & \measureD{a} \\
        & \prepareC{\beta} & \ghost{\Transformation{D}} & \measureD{b} }
    \end{aligned}
    = 
    \lambda_{a\alpha}^{\Transformation{D}} \Effect{b}
    \Transformation{B} \State{\beta}, 
  \end{align}
  where \( \lambda_{a\alpha}^{\Transformation{C}} \), \(
  \lambda_{a\alpha}^{\Transformation{D}} \) are constants which can
  depend on the choice of \( \State{\alpha} \) and \( \Effect{a} \).
  Let us now suppose that $\lambda_{a\alpha}^\Transformation{C}=0$.
  Then we have
  \begin{align*}
    \lambda_{a\alpha}^\Transformation{C}\Effect{b}\Transformation{B}\State{\beta}=
    \mu_{b\beta}^\Transformation{C}\Effect{a}\Transformation{A}\State{\alpha}=0,
  \end{align*}
  for all $\Effect b,\State\beta$. Since by hypothesis
  $\Effect{a}\Transformation{A}\State{\alpha}\neq0$, we have
  $\mu_{b\beta}^\Transformation{C}=0$ for all $\Effect b,\State\beta$, and finally
  this implies that
  \begin{equation*}
    \begin{aligned}
      \Qcircuit @C=1em @R=.7em @! R { 
        & \prepareC{\alpha} & \multigate{1}{\Transformation{C}} & \measureD{a} \\
        & \prepareC{\beta} & \ghost{\Transformation{C}} & \measureD{b} }
    \end{aligned}
    = 0,
  \end{equation*}
  for all $\Effect{a},\Effect{b},\State{\alpha},\State{\beta}$,
  namely, by Local Discriminability, $\Transformation C=0$, contrarily
  to the hypothesis. By similar arguments we can then prove that the
  coefficients $\lambda_{a\alpha}^{\Transformation{C}}$,
  $\lambda_{a\alpha}^{\Transformation{D}}$,
  $\mu_{b\beta}^{\Transformation{C}}$, and
  $\mu_{b\beta}^{\Transformation{D}}$ are all positive.

  Comparing Eq.~\eqref{eq:atom-A-C} with Eq.~\eqref{eq:atom-B-C}, and Eq.~\eqref{eq:atom-A-D} with
  Eq.~\eqref{eq:atom-A-D} one obtains:
  \begin{align*}
    & \frac{ \lambda_{a\alpha}^{\Transformation{C}} }{ \Effect{a}
      \Transformation{A} \State{\alpha} } = \frac{ \mu_{b\beta}^{\Transformation{C}}
      }{ \Effect{b} \Transformation{B} \State{\beta} }>0,
    & & \frac{ \lambda_{a\alpha}^{\Transformation{D}} }{ \Effect{a}
      \Transformation{A} \State{\alpha} } = \frac{ \mu_{b\beta}^{\Transformation{D}}
      }{ \Effect{b} \Transformation{B} \State{\beta} }>0.
  \end{align*}
  The previous relations show that all the ratios are independent of the choices of \(
  \State{\alpha} \), \( \State{\beta} \), \( \Effect{a} \), \( \Effect{b} \), i.e.~\(
  k^{\Transformation{C}} := { \lambda_{a\alpha}^{\Transformation{C}} } / { \Effect{a}
  \Transformation{A} \State{\alpha} } = { \mu_{b\beta}^{\Transformation{C}} } / { \Effect{b}
  \Transformation{B} \State{\beta} } \) and \( k^{\Transformation{D}} := {
  \lambda_{a\alpha}^{\Transformation{D}} } / { \Effect{a} \Transformation{A} \State{\alpha} } = {
  \mu_{b\beta}^{\Transformation{D}} } / { \Effect{b} \Transformation{B} \State{\beta} } \). Using
  these definitions for \( k^{\Transformation{C}} \) and \( k^{\Transformation{D}} \) in
  Eqs.~\eqref{eq:atom-A-C},\eqref{eq:atom-A-D} one gets 
  \begin{align*}
    \begin{aligned}
      \Qcircuit @C=1em @R=.7em @! R { 
        & \prepareC{\alpha} & \multigate{1}{\Transformation{C}} & \measureD{a} \\
        & \prepareC{\beta} & \ghost{\Transformation{C}} & \measureD{b} }
    \end{aligned}
    = k^{\Transformation{C}}
    \begin{aligned}
      \Qcircuit @C=1em @R=.7em @! R { 
        & \prepareC{\alpha} & \gate{\Transformation{A}} & \measureD{a} \\
        & \prepareC{\beta} & \gate{\Transformation{B}} & \measureD{b} }
    \end{aligned},\\
    \begin{aligned}
      \Qcircuit @C=1em @R=.7em @! R { 
        & \prepareC{\alpha} & \multigate{1}{\Transformation{C}} & \measureD{a} \\
        & \prepareC{\beta} & \ghost{\Transformation{C}} & \measureD{b} }
    \end{aligned}
    = k^{\Transformation{D}}
    \begin{aligned}
      \Qcircuit @C=1em @R=.7em @! R { 
        & \prepareC{\alpha} & \gate{\Transformation{A}} & \measureD{a} \\
        & \prepareC{\beta} & \gate{\Transformation{B}} & \measureD{b} }
    \end{aligned},
  \end{align*}
  for all \( \State{\alpha} \), \( \State{\beta} \), \( \Effect{a} \), \( \Effect{b} \). By Local
  Discriminability this implies \( k^{\Transformation{C}} \Transformation{A}\otimes\Transformation{B}
  = \Transformation{C} \), and \( k^{\Transformation{D}} \Transformation{A}\otimes\Transformation{B} =
  \Transformation{D} \), namely \( \Transformation{A}\otimes\Transformation{B} \) is atomic. 
\end{proof}

\begin{proposition}\label{prop:unavoidable}
  Let \( \{ \SState{\alpha}[s][t] \}_{ (s,t) \in \Gamma_n \times \Gamma_m} \subset
  \SetStates(\sn\triangleright\sm) \) the atomic states of the system \( \sn\triangleright\sm \);
  similarly let \( \{ \SState{\alpha^\prime}[s^\prime][t^\prime] \}_{ (s^\prime,t^\prime) \in
  \Gamma_{n^\prime} \times \Gamma_{m^\prime} } \subset \SetStates(\sn^\prime\triangleright\sm^\prime)
  \) the atomic states of the system \( \sn^\prime\triangleright\sm^\prime \).  Then, the atomic
  states of the composite system \( \sys x \triangleright \sys y :=
  (\sn\triangleright\sm)(\sn^\prime\triangleright\sm^\prime) \) are the elements \(
  \SState{\alpha}[s][t] \otimes \SState{\alpha^\prime}[s^\prime][t^\prime] \).
\end{proposition}
\begin{proof}
  By definition, the system \( \sys x \triangleright \sys y \) has \(
  x \times y \) atomic states, and since \( \sys x \triangleright \sys
  y = (\sn\triangleright\sm) (\sn^\prime\triangleright\sm^\prime)\) we
  have \( x \times y = n \times m \times n^\prime \times m^\prime \).
  Since the states \( \SState{\alpha}[s][t] \otimes
  \SState{\alpha^\prime}[s^\prime][t^\prime] \in \SetStates(\sys x
  \triangleright \sys y) \) are atomic (see
  Prop.~\ref{prop:otimes-is-atomic}), different from each other, and
  their cardinality is exactly \( n \times m \times n^\prime \times
  m^\prime \), we conclude that they are the atomic states
  of \( \SetStates(\sys x \triangleright \sys y) \).
\end{proof}

\begin{proposition}\label{prop:loc-admiss-iff-admiss}
A linear map
\( \Transformation{T}\in\SetTransf_{\mathbb R}(\sn\triangleright\sm,\sp\triangleright\sq) \) is admissible if and only if is
\emph{locally admissible}, i.e.~\( \Transformation{T}(\SetStates(\sn\triangleright\sm)) \subseteq
\SetStates(\sp\triangleright\sq) \).
\end{proposition}
\begin{proof}
  First, let us recall that a map \( \Transformation{T}\in\SetTransf_{\mathbb R}(\sA,\sA^\prime) \)
  is admissible if and only if \( \Transformation{T} \otimes \Transformation{I}_{\sB} (
  \SetStates(\sA\sB) ) \subseteq \SetStates(\sA^\prime\sB) \) for every system \( \sB \).
  Let us prove the equivalence for the deterministic non-causal theory in two steps.
  \paragraph{\( ( \Rightarrow ) \):} this implication is trivial and it always holds, regardless the
theory involved; i.e.~local admissibility can be derived from
    the admissibility taking the system \( \sB \) to be the trivial one \( \sI \).
    \paragraph{\( ( \Leftarrow ) \):} the linear map \(
    \Transformation{T} \in
    \SetTransf_{\mathbb{R}}(\sn\triangleright\sm,
    \sp\triangleright\sq) \) is Locally Admissible by hypothesis,
    therefore for any atomic state \( \SState{\alpha}[s][s^\prime] \in
    \SetStates(\sn\triangleright\sm) \) we have \( \Transformation{T}
    \SState{\alpha}[s][s^\prime] =
    \SState{\alpha}[][][f^{ss^\prime}][\Xi^{ss^\prime}] \in
    \SetStates(\sp\triangleright\sq) \), where \(
    f^{ss^\prime}:~\Xi^{ss^\prime}\subseteq\Gamma_p \to \Gamma_q \).
    Notice that, since for \( s_0 \neq s_1 \) the state
    $\SState{\alpha}[s_0][s_0^\prime] +
    \SState{\alpha}[s_1][s_1^\prime]$ is valid, then by Local
    Admissibility also \( \Transformation{T}
    [~\SState{\alpha}[s_0][s_0^\prime] +
    \SState{\alpha}[s_1][s_1^\prime]~] = \SState{\alpha}[][][f^{s_0
      s_0^\prime}][\Xi^{s_0 s_0^\prime}] + \SState{\alpha}[][][f^{s_1
      s_1^\prime}][\Xi^{s_1 s_1^\prime}] \) is a valid state,
    therefore we must have that
\begin{equation}\label{eq:null-intersection}
\Xi^{s_0 s_0^\prime} \cap \Xi^{s_1 s_1^\prime} = \emptyset\qquad \forall s_0^\prime,\forall
s_1^\prime\text{ and } s_0 \neq s_1.
\end{equation}

For an arbitrary  system \(
\sn^\prime\triangleright\sm^\prime \), let us choose freely the state
  \( \SState{\alpha}[][][g][\Upsilon] \) of the composite system \( \sys x \triangleright \sys y :=
  (\sn\triangleright\sm) (\sn^\prime\triangleright\sm^\prime) \). It can be expanded on the atomic
  multipartite states \( \SState{\alpha}[s][s^\prime] \otimes \SState{\alpha}[t][t^\prime] \)---with
  \( \SState{\alpha}[s][s^\prime] \in \SetStates(\sn\triangleright\sm) \), \(
  \SState{\alpha^\prime}[t][t^\prime] \in \SetStates(\sn^\prime\triangleright\sm^\prime) \)---as \(
  \SState{\alpha}[][][g][\Upsilon] = \sum_{ss^\prime tt^\prime} \alpha_{ss^\prime tt^\prime}
  \SState{\alpha}[s][s^\prime] \otimes \SState{\alpha^\prime}[t][t^\prime] \) with \( \alpha_{ss^\prime
  tt^\prime} := \delta_{s^\prime g_1(s,t)} \delta_{t^\prime g_2(s,t)} \chi_\Upsilon(s,t)
  \), for a couple of functions \( g_1:~\Upsilon \to \Gamma_m \), \( g_2:~
  \Upsilon\to \Gamma_{m^\prime} \) such that \( g(s,t) = (g_1(s,t),g_2(s,t))
  \). 
On such arbitrary multipartite state the map \( \Transformation{T} \otimes
  \Transformation{I}_{\sn^\prime\triangleright\sm^\prime} \) leads to a valid state of the composite
  system \(\sys x'\triangleright\sys y':= (\sp\triangleright\sq)(\sn^\prime\triangleright\sm^\prime) \):
  \begin{align}
    & \notag [\ \Transformation{T} \otimes
      \Transformation{I}_{\sn^\prime\triangleright\sm^\prime}\ ] \ \SState{\alpha}[][][g][\Upsilon]
      = \\ 
    & \notag = \sum_{ss^\prime tt^\prime} \alpha_{ss^\prime tt^\prime}
      \Transformation{T}{\SState{\alpha}[s][s^\prime]} \otimes
      \Transformation{I}_{\sn^\prime\triangleright\sm^\prime} \SState{\alpha}[t][t^\prime] = \\
    & \notag = \sum_{ss^\prime tt^\prime} \alpha_{ss^\prime tt^\prime}
      {\SState{\alpha}[][][f^{ss^\prime}][\Xi^{ss^\prime}]} \otimes
\SState{\alpha}[t][t^\prime] = \\
    & \notag = \sum_{ss^\prime tt^\prime vv^\prime} \alpha_{ss^\prime tt^\prime}  \delta_{v^\prime f^{ss^\prime}(v)} \chi_{\Xi^{ss^\prime}}(v)
      \SState{\alpha}[v][v^\prime] \otimes
\SState{\alpha}[t][t^\prime] = \\
%    & \notag = \sum_{ss^\prime} \sum_{tt^\prime vv^\prime} \delta_{s^\prime g_1(s,t)} \delta_{t^\prime g_2(s,t)} \chi_\Upsilon(s,t)
% \delta_{v^\prime f^{ss^\prime}(v)} \chi_{\Xi^{ss^\prime}}(v) \times \\
%& \qquad\qquad\qquad\qquad\qquad\qquad \times
%      \SState{\alpha}[v][v^\prime] \otimes
%\SState{\alpha}[t][t^\prime].\label{eq:dsa}
    & \begin{aligned}
 = \sum_{ss^\prime} \sum_{tt^\prime vv^\prime} \delta_{s^\prime g_1(s,t)} \delta_{t^\prime g_2(s,t)} \chi_\Upsilon(s,t)
 \delta_{v^\prime f^{ss^\prime}(v)} \chi_{\Xi^{ss^\prime}}(v) \times \\
      \times \SState{\alpha}[v][v^\prime] \otimes
\SState{\alpha}[t][t^\prime].\label{eq:dsa} 
\end{aligned}
  \end{align}
  The most internal sum represents the valid state \(
  \SState{\alpha}[][][h^{ss^\prime}][\Delta^{ss^\prime}] \in
  \SetStates(\sys x^\prime\triangleright\sys y^\prime) \) with \(
  h^{ss^\prime}:~\Delta^{ss'}\to \Gamma_q\times\Gamma_{m^\prime} \),
  where \( \Delta^{ss^\prime} \subseteq \Gamma_p \times
  \Gamma_{n^\prime} \) is defined by \( \chi_{\Delta^{ss^\prime}}(x,y)
  := \delta_{s^\prime g_1(s,y)} \chi_{\Upsilon}(s,y)
  \chi_{\Xi^{ss^\prime}}(x) \), and \( h^{ss^\prime}(x,y) := (
  h^{ss^\prime}_1(x,y), h^{ss^\prime}_2(x,y)) \), \(
  h_1^{ss^\prime}(x,y) := f^{ss^\prime}(x) \), \( h_2^{ss^\prime}(x,y)
  := g_2 (s,y) \). Hence the relation of Eq.~\eqref{eq:dsa} can be
  rewritten as \( [\ \Transformation{T} \otimes
  \Transformation{I}_{\sn^\prime\triangleright\sm^\prime}\ ] \
  \SState{\alpha}[][][g][\Upsilon] = \sum_{ss^\prime}
  \SState{\alpha}[h^{ss^\prime}][\Delta^{ss^\prime}] \). This sum
  represents a valid states for \( \SetStates(
  \sys x^\prime\triangleright\sys y^\prime) \) since
  the various \( \Delta^{ss^\prime} \) are disjoint: let us take two
  sets \( \Delta^{s_0 s_0^\prime} \), \( \Delta^{s_1 s_1^\prime} \),
  and evaluate \( \chi_{\Delta^{s_0 s_0^\prime}\cap\Delta^{s_1
      s_1^\prime}} \equiv \chi_{\Delta^{s_0 s_0^\prime}}(x,y)
  \chi_{\Delta^{s_1 s_1^\prime}}(x,y) \). If \( s_0 = s_1 \) we have
\begin{align*}
& \delta_{s_0^\prime g_1(s_0,y)}\delta_{s_1^\prime g_1(s_0,y)}
\chi_{\Upsilon}^2(s_0,y)
\chi_{\Xi^{s_0 s_0^\prime}}(x) 
\chi_{\Xi^{s_0 s_1^\prime}}(x) = \\
& = \delta_{s_0^\prime s_1^\prime}\delta_{s_0^\prime g_1(s_0,y)}
\chi_{\Upsilon}^2(s_0,y)
\chi_{\Xi^{s_0 s_0^\prime}}(x) 
\chi_{\Xi^{s_0 s_1^\prime}}(x) 
\end{align*}
which is equal to zero when \( s_0^\prime \neq s_1^\prime \)---thanks to the first Kronecker's
delta. On the other hand if \( s_0 \neq s_1 \) we have
\begin{align*}
& \delta_{s_0^\prime g_1(s_0,y)}\delta_{s_1^\prime g_1(s_1,y)}
\chi_{\Upsilon}(s_0,y)
\chi_{\Upsilon}(s_1,y)
\chi_{\Xi^{s_0 s_0^\prime}}(x) 
\chi_{\Xi^{s_1 s_1^\prime}}(x) \\
& \delta_{s_0^\prime g_1(s_0,y)}\delta_{s_1^\prime g_1(s_1,y)}
\chi_{\Upsilon}(s_0,y)
\chi_{\Upsilon}(s_1,y)
\chi_{\Xi^{s_0 s_0^\prime}\cap\Xi^{s_1 s_1^\prime}}(x), 
\end{align*}
which is always equals to zero thanks to
Eq.~\eqref{eq:null-intersection}, which implies $\chi_{\Xi^{s_0
    s_0^\prime}\cap\Xi^{s_1 s_1^\prime}}(x)=0$.
\end{proof}

From now on, all the admissibility proofs will be reduced to local
admissibility, thanks to Prop.~\ref{prop:loc-admiss-iff-admiss}.

\begin{proposition}\label{prop:atomic-effects}
  Under the No-Restriction Hypothesis the atomic effects of \(
  \sn\triangleright\sm \) are the elements \( \SEffect{a}[s][s^\prime]
  \) of $\SetEffects_\mathbb R(\sn\triangleright\sm)$ with \(
  (s,s^\prime) \in \Gamma_n\times\Gamma_m \) such that \( \RBraKet{
    \NSEffect{a}[s][s^\prime] }{ \NSState{\alpha}[t][t^\prime] } =
  \delta_{st} \delta_{s^\prime t^\prime} \) \( \forall
  s,t\in\Gamma_n\) and \( \forall s^\prime,t^\prime\in\Gamma_m \).
\end{proposition}
\begin{proof}
  The proof goes in three simple steps: first we show that the
  elements \( \SEffect{a}[s][s^\prime] \) are admissible. After
  showing that they are also linearly independent (therefore they span
  all the set \( \SetEffects_\mathbb{R}(\sn\triangleright\sm) \)) we
  show that every effect \( \SEffect{c} \) for the system \(
  \sn\triangleright\sm \) can be written as \( \SEffect{c} = \sum_{ij}
  c_{ss^\prime} \SEffect{a}[s][s^\prime] \) with \( c_{ss^\prime} \)
  non negative, proving that the set of atomic effects coincides with
  the set \(
  \{\SEffect{a}[s][s^\prime]\}_{(s,s^\prime)\in\Gamma_n\times\Gamma_m}
  \).

  The effects $\SEffect{a}[s][s']$ are locally admissible, since for
  every state $\SState{\alpha}[][][f][\Xi]$
  \begin{align*}
    \RBraKet{\NSEffect{a}[s][s']}{\NSState{\alpha}[][][f][\Xi]}=\sum_{tt'}\chi_{\Xi}(t)\delta_{t'f(t)}\RBraKet{\NSEffect{a}[s][s']}{\NSState{\alpha}[t][t']}=\chi_{\Xi}(s)\delta_{s'f(s)},
  \end{align*}
  which is an admissible probabilty $p\in\{0,1\}$.
  Thanks to Prop.~\ref{prop:loc-admiss-iff-admiss}, the
  $\SEffect{a}[s][s']$ are admissible, and by the No-Restriction
  Hypothesis they belong to $\SetEffects(\sn\triangleright\sm)$. 

  Now, let us show that a null linear combination of the elements \(
  \SEffect{a}[t][t^\prime] \)---say \(
  \SEffect{c}=\sum_{tt^\prime}c_{tt^\prime}\SEffect{a}[t][t^\prime]
  \)---necessarily has \( c_{tt^\prime}=0 \) \( \forall t\in\Gamma_n
  \), \( \forall t^\prime\in\Gamma_m \). Indeed, for any atomic state
  \( \SState{\alpha}[s][s^\prime] \) we get
  \begin{equation*}
    0=\RBraKet{c}{\NSState{\alpha}[s][s^\prime]} = \sum_{tt^\prime} c_{tt^\prime}
    \RBraKet{\NSEffect{a}[t][t^\prime]}{\NSState{\alpha}[s][s^\prime]} = c_{ss^\prime},
  \end{equation*}
  for every \( s,s^\prime \), i.e.~all the \( \SEffect{a}[t][t^\prime]
  \) are linearly independent.  We have that the number of different
  effects \( \SEffect{a}[t][t^\prime] \in \SetEffects(
  \sn\triangleright\sm ) \) is \( n \cdot m \), as many as \( \dim
  \SetStates_\mathbb{R}(\sn\triangleright\sm) = \dim
  \SetEffects_\mathbb{R}(\sn\triangleright\sm) = n\cdot m \): we
  conclude that the effects \( \SEffect{a}[t][t^\prime] \in
  \SetEffects( \sn\triangleright\sm )\) span the whole linear space \(
  \SetEffects_\mathbb{R}( \sn\triangleright\sm ) \).

  The third step is easily proven noticing that an arbitrary effect \(
  \SEffect{c} = \sum_{tt^\prime} c_{tt^\prime}
  \SEffect{a}[t][t^\prime] \) is a \( \{ 0,1 \} \)-functional over the
  states. Since \( \RBraKet{c}{\NSState{\alpha}[i][j]} = c_{ij}\) \(
  \forall i\in\Gamma_n \), \( \forall j \in \Gamma_m \), we conclude
  that every effect is a conic combination of the elements \(
  \SEffect{a}[t][t^\prime] \) with coefficients $0$ or $1$. Since
  linear combination with negative coefficients are forbidden we
  conclude that all the effects $\SEffect{a}[t][t^\prime]$ are atomic. For
  the same reason, there are no other atomic effects in
  $\SetEffects(\sn\triangleright\sm)$.
\end{proof}

\begin{proposition}\label{prop:effects}
  Under the No-Restriction Hypothesis the effects of the system \( \sn\triangleright\sm \) are the
  elements \( \SEffect{a}[][][v][E] := \sum_{i \in E} \SEffect{a}[v][i] \), with \( i\in\Gamma_n \),
  \( E \subseteq \Gamma_m \).
\end{proposition}
\begin{proof}
  The proof proceeds in two steps. First of all we prove that the
  elements \( \SEffect{a}[][][v][E] \in \SetEffects_\mathbb{R}(
  \sn\triangleright\sm ) \) are valid effects for the system
  $\sn\triangleright\sm$. Then we prove that there are no further
  effects in $\SetEffects(\sn\triangleright\sm)$.

  We only need to prove that the elements \( \SEffect{a}[][][v][E] \in
  \SetEffects_\mathbb{R}( \sn\triangleright\sm ) \) are locally
  admissible, and therefore they are admissible by
  Prop.~\ref{prop:loc-admiss-iff-admiss}. Finally, this implies that
  they belong to \( \SetEffects( \sn\triangleright\sm ) \) thanks to
  the No-Restriction Hypothesis.

  The effects $\SEffect{a}[][][v][E]$ are locally admissible, since
  for every state $\SState{\alpha}[][][f][\Xi]$ we have
  \begin{align*}
    \RBraKet{\NSEffect{a}[][][v][E]}{\NSState{\alpha}[][][f][\Xi]}=\chi_{E}(f(v))\chi_\Xi(v),
  \end{align*}
  which is an admissible probabilty $p\in\{0,1\}$.  

  Now let us prove that there are no other effects apart from \(
  \SEffect{a}[][][v][E] \). Given an effect \(
  \SEffect{c}\in\SetEffects(\sn\triangleright\sm) \), thanks to
  Prop.~\ref{prop:atomic-effects} we know it can be expanded over the
  atomic effects \( \SEffect{a}[t][t^\prime] \) as \( \SEffect{c} =
  \sum_{tt^\prime} c_{tt^\prime} \SEffect{a}[t][t^\prime] \) with \(
  c_{tt^\prime} = 0,1 \), \( t\in\Gamma_n \), and \(
  t^\prime\in\Gamma_m \). Suppose by contradiction that there exists a
  valid effect \( \SEffect{c} = \sum_{tt^\prime} c_{tt^\prime}
  \SEffect{a}[t][t^\prime] \) with \( c_{ij} = c_{i^\prime j^\prime} =
  1 \) for some $j$, $j^\prime$ and $i \ne i^\prime$.  Let us take the
  deterministic state \( \SState{\varepsilon}[][][f] \in
  \SetStates(\sn\triangleright\sm) \) with \( f(i)=j \) and \(
  f(i^\prime)=j^\prime \); we have that \(
  \RBraKet{c}{\NSState{\varepsilon}[][][f]} \geq 2 \), an absurd.
\end{proof}

\begin{proposition}\label{prop:atomic-are-admissible}
  Under the No-Restriction Hypothesis, the linear maps \(
  \ExtrTransformation{s}{s^\prime}{t}{t^\prime} \in \SetTransf_\mathbb{R}( \sn\triangleright\sm,
  \sp\triangleright\sq ) \) with \( (s,s^\prime,t,t^\prime) \in
  \Gamma_n\times\Gamma_m\times\Gamma_p\times\Gamma_q \) such that \(
  \ExtrTransformation{s}{s^\prime}{t}{t^\prime} \SState{\alpha}[v][v^\prime] = \delta_{sv}
  \delta_{s^\prime v^\prime} \SState{\alpha}[t][t^\prime] \), are valid transformations.
\end{proposition}
\begin{proof}
  We just need to check that the maps \(
  \ExtrTransformation{s}{s^\prime}{t}{t^\prime} \) are locally
  admissible, and then by Prop.~\ref{prop:loc-admiss-iff-admiss} and
  the No-Restriction Hypothesis, we conclude that they actually belong
  to \( \SetTransf(\sn\triangleright\sm,\sp\triangleright\sq) \).

  Indeed, for every state $\SState{\alpha}[][][f][\Xi]$, we have
  \begin{equation*}
    \ExtrTransformation{s}{s^\prime}{t}{t^\prime} \SState{\alpha}[][][f][\Xi] =   \chi_\Xi(s)\delta_{s'f(s)}\SState{\alpha}[t][t'],
  \end{equation*}
  which is a valid state of $\sp\triangleright\sq$.
\end{proof}

\begin{proposition}\label{prop:atomic-are-li}
  The transformations \( \ExtrTransformation{s}{s^\prime}{t}{t^\prime} \in \SetTransf(
  \sn\triangleright\sm, \sp\triangleright\sq ) \) are linearly  independent.
\end{proposition}
\begin{proof}
  Let us show that a null linear combination of the transformations \( \ExtrTransformation{s}{s^\prime}{t}{t^\prime} \in \SetTransf(
  \sn\triangleright\sm, \sp\triangleright\sq ) \)---say \( \Transformation{A} = \sum_{ss^\prime
tt^\prime} c_{ss^\prime tt^\prime} \ExtrTransformation{s}{s^\prime}{t}{t^\prime} \)---necessarily
has \( c_{ss^\prime tt^\prime}=0 \), for all \( s\in\Gamma_n \), \( s^\prime\in\Gamma_m \), \(
t\in\Gamma_p \), \( t^\prime\in\Gamma_q \). Indeed, for any couple \(
\SState{\alpha}[i][i^\prime] \in \SetStates( \sn\triangleright\sm  ) \), \(
\SEffect{a}[j][j^\prime] \in \SetStates( \sp\triangleright\sq  ) \) we have
\begin{equation*}
0 = \SEffect{a}[j][j^\prime] \Transformation{A} \SState{\alpha}[i][i^\prime] = c_{ii^\prime
jj^\prime},
\end{equation*}
for every \( (i,i^\prime,j,j^\prime) \in \Gamma_n \times \Gamma_m \times \Gamma_p \times \Gamma_q
\), i.e.~the transformations \( \ExtrTransformation{s}{s^\prime}{t}{t^\prime} \in \SetTransf(
  \sn\triangleright\sm, \sp\triangleright\sq ) \) are linearly independent.
\end{proof}

\begin{proposition}\label{prop:atomic-are-atomic}
  The transformations \( \ExtrTransformation{s}{s^\prime}{t}{t^\prime} \in \SetTransf(
  \sn\triangleright\sm, \sp\triangleright\sq ) \) are atomic.
\end{proposition}
\begin{proof}
  Let us suppose by contradiction that the transformation \(
  \ExtrTransformation{s}{s^\prime}{t}{t^\prime} \) is not atomic, namely \(
  \ExtrTransformation{s}{s^\prime}{t}{t^\prime} = \Transformation{A} + \Transformation{B} \) for some
  \( \Transformation{A}, \Transformation{B} \in \SetTransf( \sn\triangleright\sm, \sp\triangleright\sq
  ) \). For an arbitrary state \( \SState{\alpha}[][][f][\Xi] \in \SetStates( \sn\triangleright\sm)
  \) we have that \( \ExtrTransformation{s}{s^\prime}{t}{t^\prime} \SState{\alpha}[][][f][\Xi] =
  \chi_{\Xi}(s)\ \delta_{s^\prime f(s)}\ \SState{\alpha}[t][t^\prime] \), \( \Transformation{A}
  \SState{\alpha}[][][f][\Xi] = \SState{\alpha^{\Transformation{A}}} = \sum_{ss^\prime}
  c_{ss^\prime}^{\Transformation{A}} \SState{\alpha}[s][s^\prime] \), \( \Transformation{B}
  \SState{\alpha}[][][f][\Xi] = \SState{\alpha^{\Transformation{B}}} = \sum_{ss^\prime}
  c_{ss^\prime}^{\Transformation{B}} \SState{\alpha}[s][s^\prime] \), where we have expanded the
  states \( \SState{\alpha^{\Transformation{A}}}, \SState{\alpha^{\Transformation{B}}} \in
  \SetStates(\sp\triangleright\sq) \) over the atomic states \( \SState{\alpha}[s][s^\prime] \) of the
  system \( \sp\triangleright\sq \).  By hypothesis we have \(
  \ExtrTransformation{s}{s^\prime}{t}{t^\prime} \SState{\alpha}[][][f][\Xi] = \Transformation{A}
  \SState{\alpha}[][][f][\Xi] + \Transformation{B} \SState{\alpha}[][][f][\Xi] \), namely
  \begin{equation*}
    \chi_{\Xi}(s) \delta_{s^\prime f(s)} \SState{\alpha}[t][t^\prime] = \sum_{ss^\prime}
    c_{ss^\prime}^{\Transformation{A}} \SState{\alpha}[s][s^\prime] + \sum_{ss^\prime}
    c_{ss^\prime}^{\Transformation{B}} \SState{\alpha}[s][s^\prime]   
  \end{equation*}
  Since the atomic states \( \SState{\alpha}[s][s^\prime] \) are linearly independent we have that
  the previous relation can be rewritten as
  \begin{align*}
    & c_{ss^\prime}^{\Transformation{A}} + c_{ss^\prime}^{\Transformation{B}} = \chi_{\Xi}(s)
      \delta_{s^\prime f(s)} && \text{if }t=s,t^\prime=s^\prime \\
    & c_{ss^\prime}^{\Transformation{A}} + c_{ss^\prime}^{\Transformation{B}} = 0 &&
      \text{otherwise}
  \end{align*}
  Since \( c_{ss^\prime}^{\Transformation{A}}, c_{ss^\prime}^{\Transformation{B}} = 0,1 \), we
  conclude from the second relation that \( c_{ss^\prime}^{\Transformation{A}} =
  c_{ss^\prime}^{\Transformation{B}} = 0 \) if \( t \ne s \) or \( t^\prime \ne s^\prime \), while the
  first leads to \( c_{tt^\prime}^{\Transformation{A}} = \chi_{\Xi}(s)\ \delta_{s^\prime f(s)} \) and
  \( c_{tt^\prime}^{\Transformation{B}} = 0 \) (or the other way round). Since the initial state \(
  \SState{\alpha}[][][f][\Xi] \in \SetStates(\sn\triangleright\sm) \) is arbitrary we conclude that \(
  \ExtrTransformation{s}{s^\prime}{t}{t^\prime} = \Transformation{A} + 0 \) (or \(
  \ExtrTransformation{s}{s^\prime}{t}{t^\prime} = 0 + \Transformation{B} \)), i.e.~\(
  \ExtrTransformation{s}{s^\prime}{t}{t^\prime} \) is atomic.
\end{proof}

\begin{proposition}\label{prop:no-other-atomic}
  There are no atomic transformations in \( \SetTransf(\sn\triangleright\sm,\sp\triangleright\sq) \)
  other than \( \ExtrTransformation{s}{s^\prime}{t}{t^\prime} \).
\end{proposition}
\begin{proof}
  Since the dimension of \( \SetTransf_{\mathbb{R}}(
  \sn\triangleright\sm, \sp\triangleright\sq ) \) is \( \dim
  \SetStates_\mathbb{R}( \sn\triangleright\sm ) \times \dim
  \SetStates_\mathbb{R}( \sp\triangleright\sq ) = n \times m \times p
  \times q \), and the number of (linearly independent) atomic
  transformations \( \ExtrTransformation{s}{s^\prime}{t}{t^\prime} \)
  is \( n \times m \times p \times q \) we conclude that such atomic
  maps span the entire space of linear transformations between the two
  linear spaces of states.

  Now let us suppose by contradiction that there exists another atomic transformation \(
  \Transformation{T} \in \SetTransf(\sn\triangleright\sm,\sp\triangleright\sq) \), different from any
  of \( \ExtrTransformation{s}{s^\prime}{t}{t^\prime} \). Since the maps \(
  \ExtrTransformation{s}{s^\prime}{t}{t^\prime} \) span all the space, we expand \( \Transformation{T}
  \) over them:
  \begin{equation*}
    \Transformation{T} = \sum_{ss^\prime tt^\prime} c^{ss^\prime}_{tt^\prime}
    \ExtrTransformation{s}{s^\prime}{t}{t^\prime}.
  \end{equation*}
  Since \( \Transformation{T} \) is atomic, it has to lie out of the cone built from the
  transformations \( \ExtrTransformation{s}{s^\prime}{t}{t^\prime} \); hence at least one of the
  coefficients  \(  c^{ss^\prime}_{tt^\prime} \) is negative. Since \( \SEffect{a}[t][t^\prime]
  \Transformation{T} \SState{\alpha}[s][s^\prime] = c^{ss^\prime}_{tt^\prime} \) is a probability, we
  have that \( 0 \le c^{ss^\prime}_{tt^\prime} \le 1 \), i.e.~there are no atomic transformations
  other than \( \ExtrTransformation{s}{s^\prime}{t}{t^\prime} \).
\end{proof}

\begin{proposition}\label{prop:transf-are-admissible}
  If the No-Restriction Hypothesis holds, the transformations \( \SetTransf_\mathbb{R}(
  \sn\triangleright\sm, \sp\triangleright\sq ) \ni \Transformation{T}^{f\,g}_{\Omega} :=
  \sum_{(s^\prime,t) \in \Omega} \ExtrTransformation{f(t)}{s^\prime}{t}{g(t,s^\prime)} \) with \(
  \Omega \subseteq \Gamma_p \times \Gamma_m \), \( f:~\Gamma_p \to \Gamma_n \), and \( g:~\Gamma_p
  \times \Gamma_m \to \Gamma_q \) actually belong to \( \SetTransf( \sn\triangleright\sm,
  \sp\triangleright\sq ) \).
\end{proposition}
\begin{proof}
  By Prop.~\ref{prop:loc-admiss-iff-admiss} and the No-Restriction
  Hypothesis, we just need to show that the linear maps \(
  \Transformation{T}^{f\,g}_{\Omega} \) are locally admissible.

  For an arbitrary state $\SState{\alpha}[][][h][\Xi]$ we have
  \begin{align*}
    \Transformation{T}^{f\,g}_{\Omega}\SState{\alpha}[][][h][\Xi]&=
    \sum_{(s^\prime,t) \in \Omega} \ExtrTransformation{f(t)}{s^\prime}{t}{g(t,s^\prime)}\SState{\alpha}[][][h][\Xi]\\
    &= \sum_{s^\prime} \sum_{t} \chi_\Omega(s',t) \chi_\Xi(f(t))\delta_{s'h(f(t))}\SState{\alpha}[t][g(t,s')].
  \end{align*}
  The internal sum represents the state \( \SState{\alpha}[][][g_{s^\prime}][\Upsilon_{s^\prime}]
\in \SetStates(\sp\triangleright\sq) \)
with \( g_{s^\prime}:~\Gamma_p \to \Gamma_q \), \( g_{s^\prime}(x) := g(s^\prime,x) \) and the set
\( \Upsilon_{s^\prime} \subseteq \Gamma_p \) defined by \( \chi_{\Upsilon_{s^\prime}}(x) :=
\chi_\Omega(s',x) \chi_\Xi(f(x))\delta_{s'h(f(x))} \). The whole sum \( \sum_{s^\prime}
\SState{\alpha}[][][g_{s^\prime}][\Upsilon_{s^\prime}] \) represents a valid state of \(
\sp\triangleright\sq \), indeed for every \( s_0^\prime \neq s_1^\prime \) the sets \(
\Upsilon_{s^\prime_0} \), \( \Upsilon_{s^\prime_1} \) are disjoint since
\begin{align*}
&  \chi_{\Upsilon_{s^\prime_0} \cap \Upsilon_{s^\prime_1}}(x) =  
\chi_{\Upsilon_{s^\prime_0}}(x)
\chi_{\Upsilon_{s^\prime_1}}(x) = \\
& = \chi_\Omega(s^\prime_0,x) \chi_\Omega(s^\prime_1,x)
\chi^2_\Xi(f(x)) \delta_{s^\prime_0 h(f(x))} \delta_{s^\prime_1 h(f(x))} = \\
& = \chi_\Omega(s^\prime_0,x) \chi_\Omega(s^\prime_1,x)
\chi^2_\Xi(f(x)) \delta_{s^\prime_0 h(f(x))} \delta_{s^\prime_0 s^\prime_1},
\end{align*}
which is equal to zero when \( s^\prime_0 \ne s^\prime_1 \) thanks to the last Kronecker's delta.
\end{proof}

\begin{proposition}\label{prop:transformations}
  All the elements of \( \SetTransf( \sn\triangleright\sm, \sp\triangleright\sq ) \) have
  necessarily the form: \( \Transformation{T}^{f\,g}_{\Omega} := \sum_{ (s^\prime,t) \in \Omega}
  \ExtrTransformation{f(t)}{s^\prime}{t}{g(t,s^\prime)} \) with \( \Omega \subseteq \Gamma_p \times
  \Gamma_m \), \( f:~\Gamma_p\to\Gamma_n \), and \( g:~ \Gamma_p\times\Gamma_m\to\Gamma_q \).
\end{proposition}
\begin{proof}
  Given a generic transformation \( \Transformation{T} = \sum_{ss^\prime tt^\prime} c_{ss^\prime
  tt^\prime} \ExtrTransformation{s}{s^\prime}{t}{t^\prime} \), we have that \(
  \SEffect{a}[j][j^\prime] \Transformation{T} \SState{\alpha}[i][i^\prime] =  c_{ii^\prime jj^\prime}
  \). Since \( c_{ii^\prime jj^\prime} \) is a probability in a deterministic theory, we have \( (
  c_{ss^\prime tt^\prime} = 0) \vee ( c_{ss^\prime tt^\prime} = 1 ) \), \( \forall
  (s,s^\prime,t,t^\prime) \in \Gamma_n \times \Gamma_m \times \Gamma_p \times \Gamma_q \).

  By contradiction, let us suppose that the transformation \( \Transformation{T} = \sum_{ss^\prime
  tt^\prime} c_{ss^\prime tt^\prime}  \ExtrTransformation{s}{s^\prime}{t}{t^\prime} \) with \(
  c_{ii^\prime jj^\prime} = c_{ki^\prime jl^\prime} = 1 \) with \( i \ne k \), \( j^\prime \ne
  l^\prime \). Let \( h:\Gamma_n\to\Gamma_m \) with \( h(x)=i^\prime \) \( \forall x\in\Gamma_n \),
  then we have \( \SEffect{e}[][][j] \Transformation{T} \SState{\varepsilon}[][][h] \geq 2 \), i.e.~an
  absurd.

  In such a way we have not ruled out the case \( \sum_{(s^\prime,t)\in\Omega}
  \ExtrTransformation{f(t,s^\prime)}{s^\prime}{t}{g(t,s^\prime)} \).  A transformation of this last
  form must have a couple of coefficients such that \( c_{ii^\prime jj^\prime} = c_{kk^\prime
  jl^\prime} = 1 \) with \( i \ne k \),  \( i^\prime \ne k^\prime \), \( j^\prime \ne l^\prime \),
  otherwise the functional dependence of $f$ on the variable $s^\prime$ would be trivial.  Let \(
  h:\Gamma_n\to\Gamma_m \) with \(h(i)=i^\prime \), \(h(k)=k^\prime \); then we have \(
  \SEffect{e}[][][j] \Transformation{T} \SState{\varepsilon}[][][h] \geq 2 \), i.e.~again an absurd.
\end{proof}
\section*{Bibliography}
\bibliographystyle{elsarticle-num}
\bibliography{determinism_without_causality}
\end{document}